\documentclass[conference,compsoc]{IEEEtran}

\usepackage{mathtools}
\usepackage{graphicx}
\usepackage{amsmath}
\usepackage{amsthm}
\usepackage{amssymb}
\usepackage{enumerate}
\usepackage{cite}
\usepackage[numbers,sort&compress]{natbib}
\usepackage[]{algorithm2e}

\theoremstyle{plain}
\newtheorem{thm}{Theorem}[section]
\newtheorem{lem}[thm]{Lemma}

\theoremstyle{definition}
\newtheorem{defn}{Definition}[section]

\ifCLASSINFOpdf
\else
\fi
\begin{document}
\title{Forwarding and Optical Indices in an All-Optical BCube Network}

\author{\IEEEauthorblockN{Suzhen Wang, Jingjing Luo, Wing Shing Wong}
\IEEEauthorblockA{Department of Information Engineering\\
The Chinese University of Hong Kong\\
Email: {ws012,  jjluo,  wswong}@ie.cuhk.edu}
\and
\IEEEauthorblockN{Yuan-Hsun Lo}
\IEEEauthorblockA{The School of Mathematical Sciences\\
Xiamen University, Xiamen 361005, China\\
Email: yhlo0830@gmail.com}
}
\maketitle

\begin{abstract}
BCube is a highly scalable and cost-effective networking topology, which has been widely applied to modular datacenters.
Optical technologies based on Wavelength Division Multiplexing (WDM) are gaining popularity for Data Center Networks (DCNs) due to their technological strengths such as low communication latency, low power consumption, and high link bandwidth. 
Therefore, it is worth investigating optical techniques into the BCube architecture for future DCNs.
For this purpose, we study the forwarding and optical indices in an all-optical BCube network. Consider an all-optical BCube network in which every host sets up a connection with every other host. The optical index is the minimum number of wavelengths required by the network to support such a host-to-host traffic, under the restriction that each connection is assigned a wavelength that remains constant in the network. A routing is a set of directed paths specified for all host pairs. By defining the maximum link load of a routing as the maximum number of paths passing through any link, the forwarding index is measured to be the minimum of maximum link load over all possible routings. The forwarding index turns out to be a natural lower bound of the optical index. In this paper, we first compute the forwarding index of an all-optical BCube network. Then, we derive an upper bound of the optical index by providing an oblivious routing and wavelength assignment (RWA) schemes, which attains the lower bound given by the forwarding index in some small cases. Finally, a tighter upper bound is obtained by means of the chromatic numbers in Graph Theory.
\end{abstract}

\IEEEpeerreviewmaketitle

\section{Introduction}

Data Center Networks  (DCNs) are core infrastructures for various online services and cloud applications such as social networking and cloud computing.
Those services and applications are engendering an exponential traffic growth, which places a significant demand on network bandwidth. 
To meet this demand, all-optical DCNs arise as promising architectures because they offer extremely high bandwidth by adopting Wavelength Division Multiplexing (WDM) \cite{all_optical_scheduling, wave_Cube_ToN, OSA_optical_switching}. Besides, optical DCNs are reported to consume much less power compared with electronic DCNs  \cite{Helios,OSA_optical_switching}. 

In a large-scale datacenter deployment, traditional hierarchical tree topologies face issues such as link oversubscription and network bisection-bandwidth bottlenecks. 
To address these issues, researchers have proposed various scalable topology solutions such as  fat-trees \cite{al2008scalable}, BCube \cite{guo2009bcube}, and ExCCC \cite{ExCCC-DCN}. 
A fat-tree is a folded version of a Clos network which was first designed in mid-1950s \cite{clos1953study}. 
BCube, as a modified version of Hypercube, was proposed recently by  Guo \textit{et al.}  \cite{guo2009bcube} for building modular datacenters.
Both fat-trees and BCube achieve linear relationships between the network  bisection bandwidth and the  network size. 
However, BCube is reported to be more cost-effective than fat-trees.
In addition, a BCube network can push the routing and scheduling functionalities to end-servers, which helps alleviate the routing burden on intermediate switches. 
For simplicity, we refer to servers or end-servers as hosts.

Note that all-optical DCNs are promising and BCube is  highly scalable and economic for building modular datacenters. 
In this paper, we make the first attempt to study the fundamental problem of Routing and Wavelength Assignment (RWA) in an all-optical BCube network considering a \textit{host-to-host} traffic, where we assign every Source-Destination (S-D) host pair with a nonblocking \textit{lightpath} --- consists of a single physical path and a single wavelength --- such that all host pairs can communicate simultaneously.
We describe lightpaths are \textit{nonblocking} if lightpaths that share a common link have different wavelengths.
Since wavelength is a limited resource, the goal of  the RWA problem is to minimize wavelength usage \cite{routing_W, RWA_opt}. 
Although a host-to-host traffic may not arise frequently in practice, it evaluates the maximum communication capacity of a network and also locates a reference point for further communication analysis.  
To simplify the analysis, we divide the RWA problem into two parts: path allocation and wavelength assignment. 
In the part of path allocation, we aim to find a set of dipaths that minimizes the maximum link load; 
in the part of wavelength assignment, we aim to minimize the usage of wavelengths.
Specifically, the minimum of the maximum link load over all possible routings is referred to as the \textit{forwarding index} when link load is measured by the number of paths passing through it \cite{optical_fowarding, forward_optical_paper}.
We refer to the minimum number of wavelengths, used to support simultaneous host-to-host communication, as the \textit{optical index}. 
It has been shown that the optical index is naturally lower bounded by the forwarding index \cite{forward_optical_paper, optical_fowarding}.

It is NP-hard to derive either the optical index or the forwarding index in a general network since these problems are shown to be more complicated than a vertex coloring problem \cite{full_connections}. 
Therefore, there have been numerous attempts to study the optical and forwarding index in various interconnection networks such as fat-trees \cite{global_packing}, $4$-regular circulant networks \cite{forward_optical_paper}, and some Cartesian product of chains or cycles  \cite{full_connections, all-to-alloptical}. 
In particular, Lo \textit{et al.} \cite{global_packing} derived the optical index in an all-optical fat-tree network through explicit construction of an RWA scheme;  Beauquier \cite{full_connections} derived the forward and optical indices for some Cartesian product of simple graphs, such as cycles, chains and complete graphs. 
In this paper, we report three results shown as follows. 
First, we derive the value of the \textit{ fowarding index} in a BCube network. 
Second, we propose an \textit{oblivious} RWA scheme; the term oblivious signifies that the RWA assigns a lightpath to an S-D pair based only on its source and destination addresses.
Third, we derive an upper bound and a lower bound of the optical index in a BCube network.
The derived results can provide  insights into optimal lightpath allocation, and serve as a baseline for future research in more sophisticated RWA schemes in BCube networks. 

The rest of the paper is organized as follows.
Section 2 introduces some preliminaries.  
Section 3 introduces the BCube topology. 
Section 4 presents analysis on host-to-host communication in a BCube network. 
We conclude the paper in section 5. 

\section{Preliminaries}
In this paper, we consider a full-duplex network, where each node can send and receive messages at the same time.  
Hence, we can model an all-optical network by a \textit{symmetric digraph} --- a directed graph $G $ with vertex set $V(G)$ and arc set $A(G)$ such that if  $\alpha_{x,  y} \in A(G)$ then $\alpha_{y, x}\in A(G)$. 
Here $\alpha_{x, y}$ represents an arc directed from node $x$ to node $y$.
Let $P_{s, d}$ denote a directed path (dipath) from source node $s$ to destination node $d$. 
A set of dipaths is called routing. For a given routing $R$ of $G$, let 
$\pi(G, R, \alpha_{x, y})$ denote the load of arc $\alpha_{x,y}$ with respect to $R$. 
which is measured by the number of dipaths in $R$ that pass through $\alpha_{x,y}$. 
The maximum link load is then denoted by $\pi(G, R) := \max_{\alpha_{x, y} \in A(G)} \pi(G, R, \alpha_{x, y})$. 
Let $\mathcal{R}$ denote the collection of all possible routings.
Then the \textit{forwarding index} of a graph $G$, denoted by $\pi(G)$, is defined as the minimum of the maximum link load over all routings, i.e.,
\begin{equation}
\pi(G) :=\min_{R\in \mathcal{R}} \pi(G, R).
\end{equation}
To study the RWA problem, we represent wavelengths by different colors.
In an optimal wavelength assignment, the number of colors required is minimal.
Let $\omega(G, R)$ denote the minimum number of colors required to color dipaths of $R$ such that dipaths  are assigned with different colors if they share a common arc. 
The \textit{optical index} in a graph of $G$, denoted by $\omega(G)$, is defined as
\begin{equation}
\omega(G) = \min_{R \in \mathcal{R}} \omega(G, R).
\end{equation}
Since dipaths sharing a common arc should be assigned with different colors, we have 
\begin{equation}
\label{inequality_1}
\omega(G)\ge \pi(G).  \quad \quad \quad
\end{equation}

It is difficult to investigate whether the equality in (\ref{inequality_1}) holds for a general topology.
However, researchers have shown that the equality holds for some specific topologies such as cycles \cite{full_connections}, Hypercubes \cite{full_connections}, trees of cycles \cite{all-to-all3}, some Cartesian product of paths or cycles with equal length \cite{full_connections, all-to-alloptical}, and some circulant graphs \cite{forward_optical_paper}. 

In this paper, we evaluate the forwarding and optical indices in a BCube network by considering host-to-host routings. 
More precisely, by denoting $V_h$ the set of hosts in a BCube, a host-to-host routing is given by
$R  = \{P_{s, d}: s, d\in V_h(G), s\neq d\}$. 
The structure of BCubes will be given in next section.

\begin{figure}[htbp]
\begin{center}
\includegraphics[width = 3.2 in]{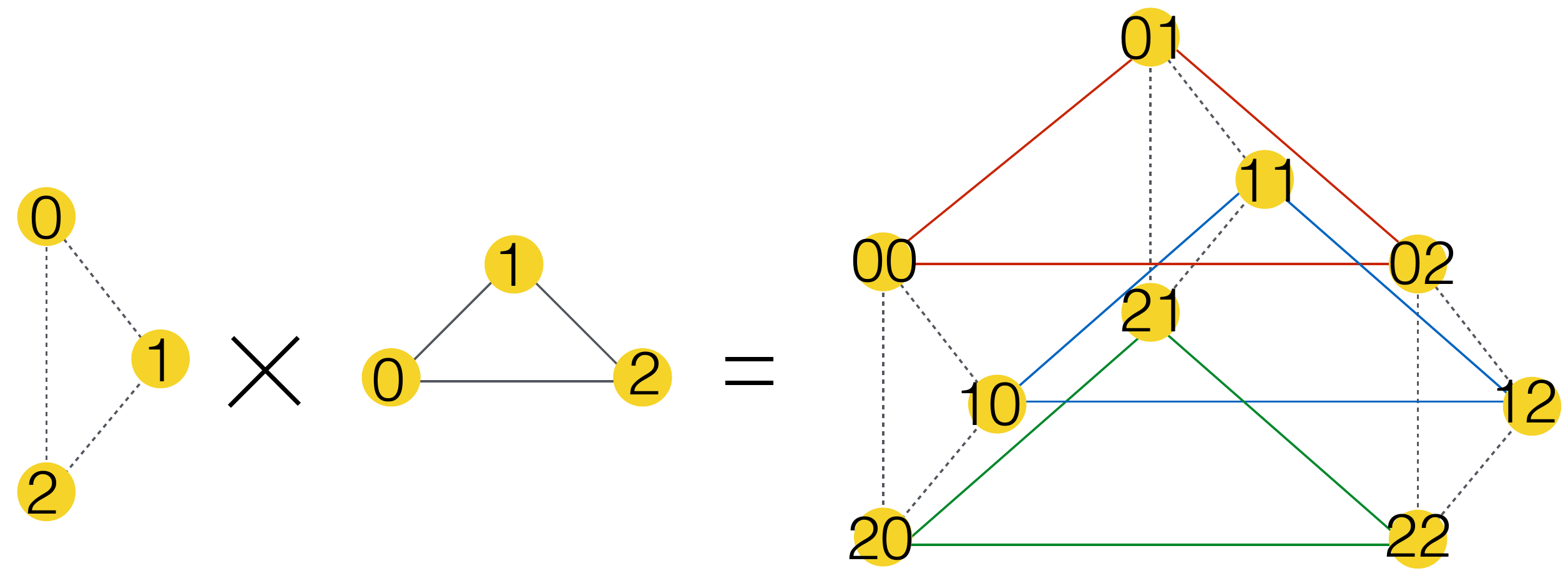}
\caption{ $\mathcal{H}(3, 3) = K_3 \times K_3$}
\label{product}
\end{center}
\end{figure}

BCube is closely related to the generalized Hypercube. Towards a better understanding of BCube, we first briefly explain some properties of Hypercube.
Let $K_{n}$ denote a complete graph with $n$ nodes 
indexed by integers in $\mathbb{Z}_{n}$, where $\mathbb{Z}_{n} := \{0, 1, ..., n-1 \}$.
Any two nodes in $K_{n}$ are adjacent to each other.
Since a generalized $\ell$-dimensional Hypercube $\mathcal{H}(n_1, ..., n_\ell)$ is the Cartesian product of complete graphs $K_{n_i},  i= 1, 2, ..., \ell$, we have
\begin{displaymath}
\mathcal{H}(n_1, ..., n_\ell) : = K_{n_1} \times ... \times K_{n_\ell}.
\end{displaymath} 
The node set $V(\mathcal{H}(n_1, ..., n_\ell))$ is $\{(v_1,\ldots,v_n): \ v_i \in V(K_{n_i})\}$, where each node of $\mathcal{H}(n_1, ..., n_\ell)$  can be expressed by an $\ell$-dimensional vector, $\textbf{h} =h_1 ... h_\ell \in \mathbb{Z}_{n_1}\times ... \times \mathbb{Z}_{n_\ell}$. 
Two nodes in $\mathcal{H}(n_1, ..., n_\ell)$ are adjacent if their vectors differ only in one component.
Fig. \ref{product} illustrates $\mathcal{H}(3, 3)$, which is the Cartesian product of $K_3$ and $K_3$.
Readers can refer to \cite{bhuyan1984generalized} for more details on Cartesian product of graphs and Hypercube. It has been shown in \cite{full_connections} that $\omega(\mathcal{H}(n_1, ..., n_\ell) )= d^{\ell}-d^{\ell-1}$ when $n_i= d $ for any $i\in \{1, 2, ..., \ell\}$.
\begin{figure}[htbp]
\begin{center}
\includegraphics[width = 3.2 in]{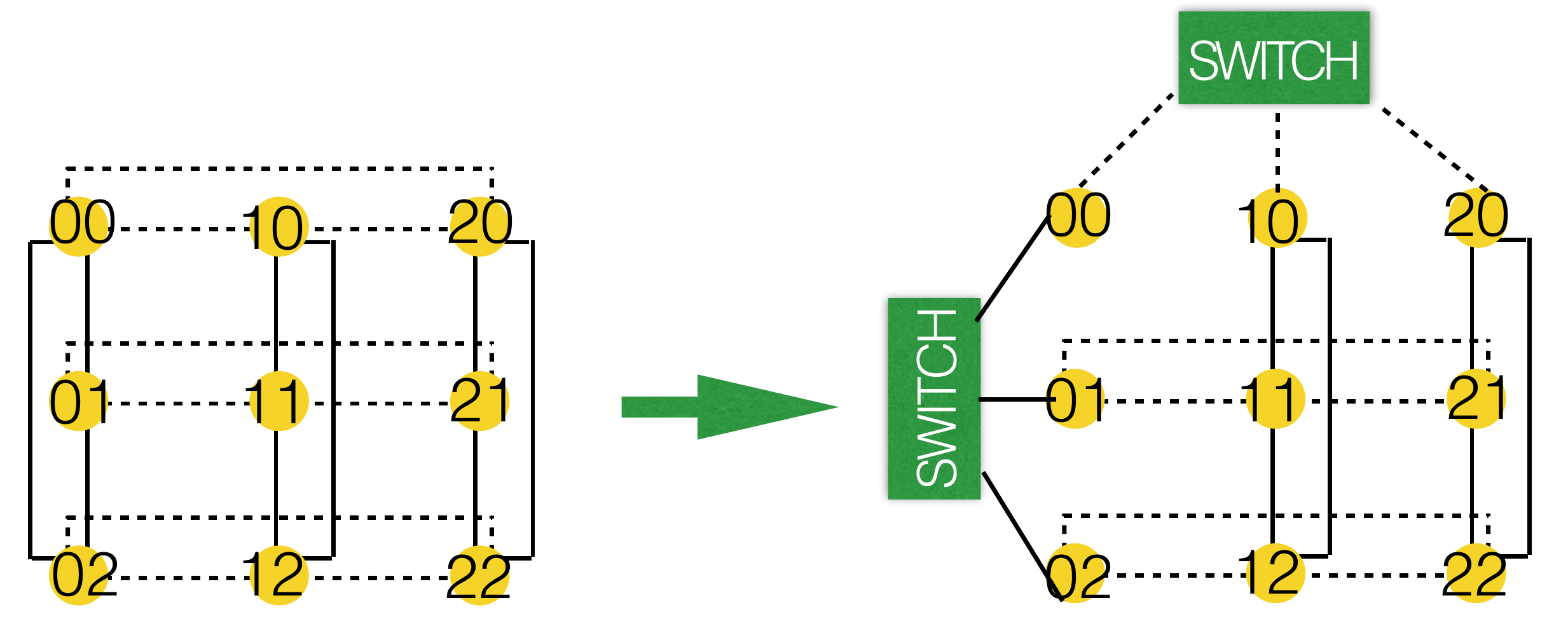}
\caption{An intermediate state of transforming $\mathcal{H}(3, 3)$ into $\mathcal{B}(2, 3)$}
\label{hypercube-to-BCube}
\end{center}
\end{figure}

The main difference between Hypercube and BCube lies in how adjacent nodes are connected to each other. 
In Hypercube adjacent nodes are connected directly by edges; however, in BCube they are connected via common switches.
Such a difference contributes to a reduction in wiring complexity for building large-scale networks. 
Fig. \ref{hypercube-to-BCube} shows an intermediate process of transforming Hypercube into BCube, where each torus is replaced by a switch.

\section{ The BCube Topology}
We use $\mathcal{B}(\ell, d)$ to denote a BCube network which has one host layer and $\ell$ switch layers; this network is constructed by optical $d$-port switches. 
We  index switch layers from $1$ to $\ell$  from bottom to top, and index ports in a switch from $0$ to $d-1$ from left to right. Since we consider full-duplex networks, we assume these ports are bidirectional.
Similar to Hypercube, we denote hosts in $\mathcal{B}(\ell, d)$ by $\ell$-dimensional vectors, $\textbf{h} = h_1... h_\ell \in \mathbb{Z}_d^\ell$, and we denote switches in $\mathcal{B}(\ell, d)$ by $(\ell-1)$-dimensional vectors, $\textbf{s}^k = s^k_1 ... s^k_{\ell-1} \in \mathbb{Z}^{\ell-1}_d$. Here $k$ indicates a switch at the $k$-th switch layer. Hereinafter, we simply use layer to refer to switch layer. 
Fig. \ref{fig.0} and Fig. \ref{fig.1} illustrate the structures of $\mathcal{B}(2, 3)$ and $\mathcal{B}(3, 3)$, respectively. 
Since  $\mathcal{B}(\ell, d)$ consists of switches and hosts, we have  $V(\mathcal{B}(\ell, d)) = V_s(\mathcal{B}(\ell, d)) \cup V_h(\mathcal{B}(\ell, d))$, where $V_s(\mathcal{B}(\ell, d))$ is the switch set and $V_h(\mathcal{B}(\ell, d))$ is the host set. 
By letting $N_h$, $N_s$ and $N_\alpha$ be the number of hosts, switches and arcs, we have $N_h = d^\ell$ and  $N_\alpha = 2\ell N_h$ since each host has $\ell$ bidirectional links. 
\begin{figure}[ht]
\centering
\includegraphics[width = 3.2 in]{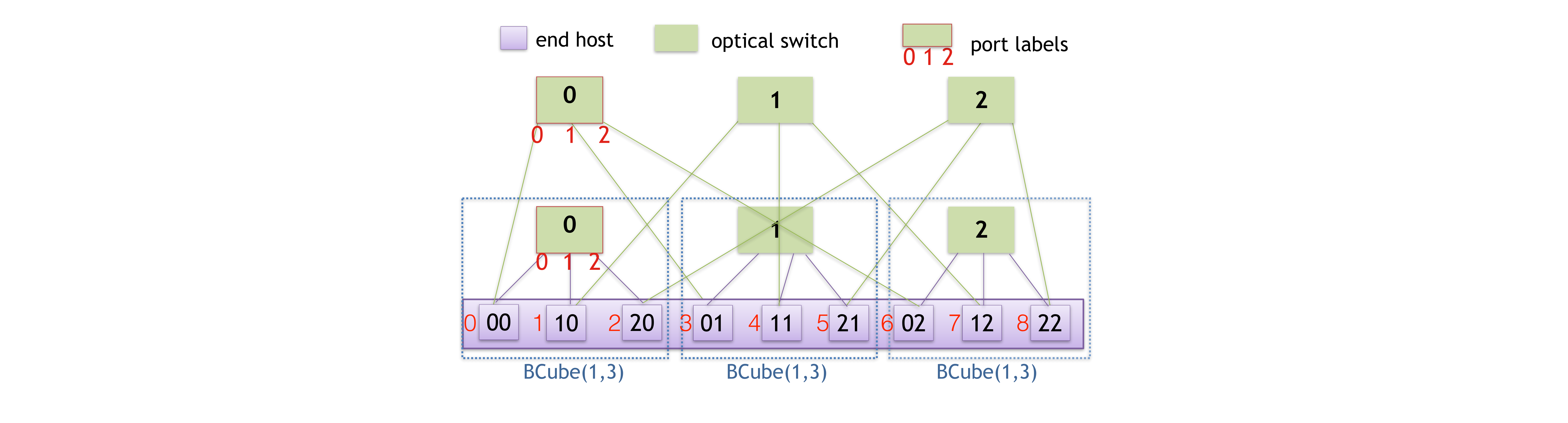}
\caption[Recursive Construction:]
 {The bottom layer in purple is the host layer; the remaining layers are switch layers. This BCube has three built-in $\mathcal{B}(1, 3)$, each of which is in a blue dashed rectangle.
 }
   \label{fig.0}
\end{figure}

\noindent \textbf{Recursive construction:} 
$\mathcal{B}(1, d)$ is constructed by $d$ hosts and one switch, where these hosts are all connected directly to this switch. 
For $ \ell > 1$,  $\mathcal{B}(\ell, d)$ is constructed by $d$ $\mathcal{B}(\ell-1, d)$, where hosts in different  $\mathcal{B}(\ell-1, d)$ are connected by switches at the $\ell$-th layer.

In particular, we refer to these $\mathcal{B}(\ell-1, d)$ as \textit{built-in} BCubes of $\mathcal{B}(\ell, d)$ since they are inside $\mathcal{B}(\ell, d)$.
For example, the three $\mathcal{B}(2, 3)$ in Fig. \ref{fig.1} are built-in BCubes of $\mathcal{B}(3, 3)$.
Accordingly, we can split a host vector in $\mathcal{B}(\ell, d)$ into two parts, i.e., $\textbf{h}  = h_{1:\ell-1} h_{\ell}$, where $h_{1:\ell-1} = h_1...h_{\ell-1}$ is equal to a host vector in $\mathcal{B}(\ell-1, d)$, and $h_{\ell}$ can identify different built-in $\mathcal{B}(\ell-1, d)$.
In particular, $h_\ell$ is the same for all hosts belonging to the same build-in BCube.
For simplicity, we let $h_{\ell}$ be the index of a built-in $\mathcal{B}(\ell-1, d)$.
For example, the indices of  three built-in $\mathcal{B}(2,3)$ from left to right in Fig. \ref{fig.1} are $0$, $1$, and $2$, respectively. 

In $\mathcal{B}(\ell, d)$, a link exists only between a host and a switch; a host is physically connected to $\ell$ switches at different layers; and a switch is physically connected to $d$ hosts. 
We describe two hosts are \textit{neighbors} if their vectors differ in only one component.
If two neighbor hosts differ in the $k$-th components, the two hosts are connected directly to a common switch at the $k$-th layer. 
Mathematically, a physical link exits between a host $\textbf{h}$ and a switch $\textbf{s}^k$ if and only if the following equation is satisfied. 
\begin{equation}
\label{connection_def}
s^k_1... s^k_{\ell -1} = h_1 ... h_{k-1} h_{k+1} ... h_\ell.
\end{equation}
If a host is connected to the $h_k$-th port of switch $\textbf{s}^k$ via a link,  we can infer the vector of this host according to (\ref{connection_def}).
For example,  in Fig. \ref{fig.1}, switch $\textbf{s}^3 = 20$ is physically connected to hosts $2 0 \underline{0}$, $2 0 \underline{1}$ and $2 0 \underline{2}$ via its ports $0$, $1$, and $2$, respectively; the underlined numbers are determined by the corresponding port indices.
Moreover, switch $\textbf{s}^2 =2 0$ is directly connected to hosts $2 \underline{0} 0$, $2  \underline{1} 0$, and $2 \underline{2} 0$ via its ports $0, 1$, and $ 2$, respectively; switch $\textbf{s}^1= 2 0$ at the 1st layer is directly connected to hosts $\underline{0} 2 0$, $\underline{1}  2  0$ and $\underline{2} 2 0$ via its ports $0$, $1$, and $2$, respectively.

A directed link is referred to as an \textit{uplink} if its direction is from a host to a switch and is referred to as a \textit{downlink} if its direction is from a switch to a host. 
We use $\textbf{u}^k = u^k_1 ... u^k_\ell$ and $\textbf{d}^k=d^k_1 ... d^k_\ell$ to denote an uplink and a downlink, respectively. Here $k$ indicates that the corresponding directed link is connected to a switch at layer $k$. Note that the value of  $\textbf{u}^k $ (or $\textbf{d}^k $) is determined by the vector of its connected host. Hence, we have $\textbf{u}^k= \textbf{h}$ ($\textbf{d}^k = \textbf{h}$) if uplink $\textbf{u}^k$ (downlink $\textbf{d}^k$) is connected to host $\textbf{h}$. 

\begin{figure*}[ht]
 \centering
 \includegraphics[width= 7.2 in]{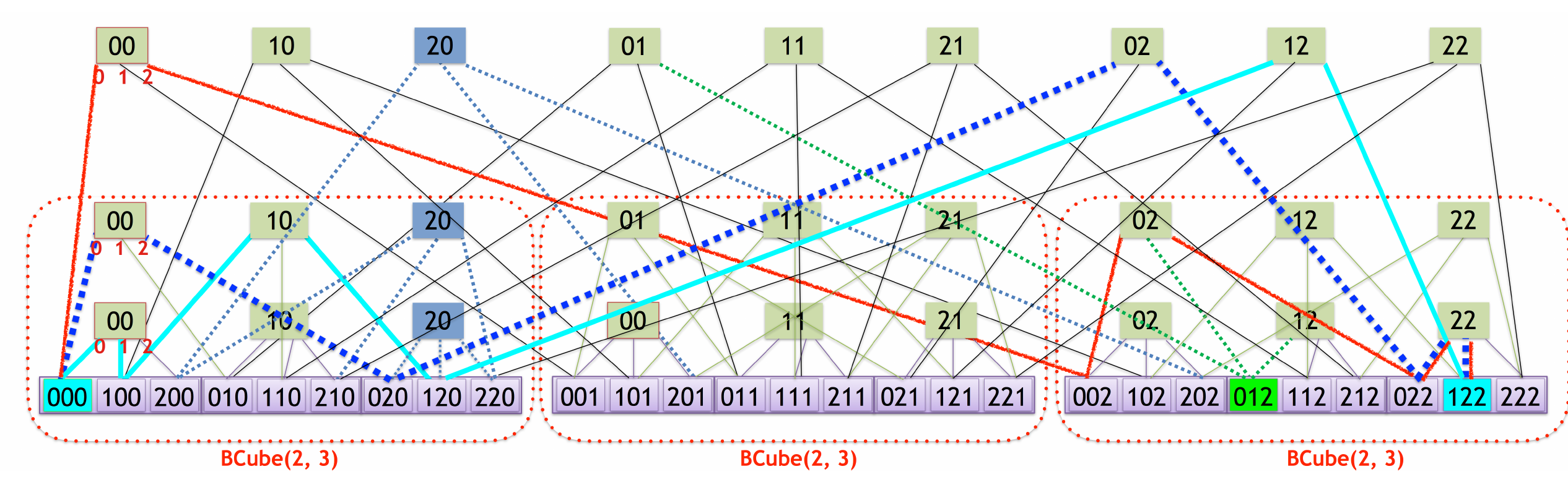}
 \caption[Topology of $ \mathcal{B}(3, 3)$]
 {$\mathcal{B}(3, 3)$ is constructed by three $\mathcal{B}(2, 3)$. Purple and blue rectangles represent  hosts and switches, respectively.}
 \label{fig.1}
\end{figure*}

\noindent \textbf{Routing in BCube:}
For an S-D host pair ($\textbf{h}^s, \textbf{h}^d$), we use the Hamming distance  $\| \textbf{h}^s - \textbf{h}^d\|_0 $ to measure the distance between $\textbf{h}^s$ and $\textbf{h}^d$.
If the Hamming distance of pair $(\textbf{h}^s, \textbf{h}^d)$ is $m$,  any of its shortest dipaths has $m$ \textit{hops}. 
A hop is defined here as from a host to a switch, and then back to a host.
Note that neighbor hosts in a BCube network can be reached through one hop. 
Let $\textbf{h}\to \textbf{h}^\prime$ denote the hop from host $\textbf{h}$ to host $\textbf{h}^\prime$. 
Then, we can represent a shortest dipath of $m$ hops from host $\textbf{h}^s$ to host $\textbf{h}^d$ as follows. 
\begin{displaymath}
\textbf{h}^s \to \textbf{h}^1 \to ...  \to \textbf{h}^{m-1} \to \textbf{h}^d.
\end{displaymath}
Each hop in a shortest dipath corresponds to one different component between $\textbf{h}^s$ and $ \textbf{h}^d$. 
In particular, if  $\textbf{h}^s$ and $ \textbf{h}^d$ differs in the $k$-th component, there exits a hop that shall traverse a switch at layer $k$; thus, we describe the hop fixes the $k$-th component. 
We define \textit{component-fixing order} as the order of fixing different components by a sequence of hops. 
For an S-D pair of distance $m$, there are $m!$ different component-fixing order, where each order uniquely determines a shortest path. 
In view of this, we conclude that BCube maintains a high degree of path diversity.  
For example, in Fig. \ref{fig.1}, S-D pair $(0 0  0, 1 2 2 )$ has the following six shortest dipaths with each corresponding to an unique component-fixing order.
\begin{equation}
\nonumber
\begin{cases}
0 0 0 \to 1 0 0 \to 1 2 0 \to 1 2 2,\\
0 0 0 \to 1 0 0 \to 1 0 2 \to 1 2 2, \\
0 0 0 \to 0 2 0 \to 1 2 0 \to 1 2 2, \\
0 0 0 \to 0 2 0 \to 0 2 2 \to 1 2 2, \\
0 0 0 \to 0 0 2 \to 1 0 2 \to 1 2 2, \\
0 0 0 \to 0 0 2 \to 0 2 2 \to 1 2 2, 
\end{cases}
\end{equation}
In particular,  dipaths that  follow the descending component-fixing order are called \textit{descending} dipaths.
For example, the descending dipath of  $(0 0 0, 1 2 2 )$ is $0 0 0 \to 0 0 2 \to 0 2 2 \to 1 2 2$.
Without loss of generality, we conduct the analysis using descending dipaths in the remaining of this paper.
\begin{defn}
For any positive integers $\ell$ and $d$, let $R^{*}(\ell,d)$ denote the host-to-host routing in $\mathcal{B}(\ell, d)$ where all involved dipaths are descending.
\end{defn}

\section{Forwarding and Optical Indices}
In this section, we first derive the exact value of $\pi(\mathcal{B}(\ell, d))$.
Then we propose an oblivious RWA for a $\mathcal{B}(\ell, d))$.
Finally, we derive the upper and lower bound of $\omega(\mathcal{B}(\ell, d))$. 
We divide the analysis of the host-to-host communication into two parts: path allocation and wavelength assignment. 
In the part of path allocation, we aim to find a set of dipaths that minimizes the maximum link load; 
in the part of wavelength assignment, we aim to minimize the usage of wavelengths.

\subsection{Forwarding Index}
To begin with, we first present an existing result on the forwarding index of a graph (see e.g., \cite{forwarding_bound}). 
\begin{lem}
\label{lem1}
For a given $G= (V, A)$, we have
\begin{displaymath}
 \pi(G) \ge \frac{N_v ( N_v -1) \bar{d}(G)}{N_\alpha},
\end{displaymath}
 where
\begin{displaymath}
\bar{d}(G) = \frac{1}{N_v(N_v-1)}\sum_{x, y\in V, x\neq y}d_{x, y}.
\end{displaymath}
Here $N_v$ and $N_\alpha$ refer to  the number of nodes and arcs, respectively, and $d_{x, y}$ denotes the distance between node $x$ and node $y$. Note that $\bar{d}(G)$ denotes the average distance over all nodes.
\end{lem}
Recall that nodes in BCube have two types: switches and hosts, and we only assign dipaths to host pairs.
In order to derive a lower bound of $\pi(\mathcal{B}(\ell, d))$ based on Lemma \ref{lem1}, we replace $N_v$ with $N_h$ which is  the number of hosts in $\mathcal{B}(\ell, d)$.
Besides, we should use the average distance only over host nodes (average host distance for short).  
Note that the average host distance is  given by 
\begin{equation}
\label{avg_distance}
\bar{d}(\mathcal{B}(\ell, d)) = 2\ell \frac{d-1}{d}\frac{N_h}{N_h-1},
\end{equation}
according to Theorem 4 of  \cite{guo2009bcube}.
We then have  
\begin{equation}
\label{forwarding_lower}
\pi(\mathcal{B}(\ell, d)) \ge  \frac{N_h(N_h-1) \bar{d}(\mathcal{B}(\ell, d))}{N_\alpha}  =  d^{\ell} - d^{\ell-1}.
\end{equation}
In what follows, we shall show $d^\ell-d^{\ell-1}$ is also an upper bound of $\pi(\mathcal{B}(\ell, d))$.

\begin{thm}
\label{thm1}
For any positive integers $\ell$ and $d$, one has 
\begin{displaymath}
\pi (\mathcal{B}(\ell, d)) =  d^{\ell} - d^{\ell-1}.
\end{displaymath}
\end{thm}
\begin{proof}
We proceed by induction on $\ell$. The steps are as follows. 
\begin{enumerate}[(1)]
\item We first show $\pi(\mathcal{B}(1, d)) = d-1$;
\item we \textit{assume} $\pi(\mathcal{B}(k, d )) = d^k - d^{k-1}$;
\item we \textit{prove} $\pi(\mathcal{B}(k+1, d)) =d^{k+1} - d^{k}$. 
\end{enumerate}
Note that  $\mathcal{B}(k+1, d)$ contains $d$ built-in $\mathcal{B}(k,d)$.

In $\mathcal{B}(1,d)$, there are $d$ hosts and only one switch.  Since every host plays as the role of source and destination exactly $d-1$ times, and the shortest dipath for an S-D pair is unique, we clearly have $\pi(\mathcal(1,d)) \leq d-1$.
Together with (\ref{forwarding_lower}), we have $\pi(\mathcal{B}(1, d))= d-1$.
Next, we assume the result holds for  $\ell=k$ and then prove the result for $\ell=k+1$.
 We first show that the maximum link load in layer  (k+1) is equal to $d^{k+1} -d^k$;  we then show that the maximum link load in each built-in BCube is equal to $d^{k+1} -d^k $. 
Consider an arbitrary uplink  $\textbf{u}^{k+1}$. 
If this uplink is traversed by the dipath of an S-D pair $(\textbf{h}^s, \textbf{h}^d)$ in $R^*(k+1, d)$, we have
\begin{equation}
\label{uplink_traverse_condition}
\textbf{h}^s =  \textbf{u}^{k+1}\  \textup{and}  \ h^d_{k+1}\neq h^s_{k+1}.
\end{equation}
Furthermore, the number of satisfied S-D pairs in $R^*(k+1, d)$ is $d^k(d-1) = d^{k+1} -d^k$ S-D pairs due to the following facts.
(1) Uplink  $\textbf{u}^{k+1}$ determines the source of these pairs. (2) Any of the first $k$ components of  $\textbf{h}^d$ has $d$ choices, whereas the last component has $d-1$ choices. 
In other words, the load of uplink $\textbf{u}^{k+1}$ is $d^{k+1} -d^k$. 

Similarly, consider an arbitrary downlink $\textbf{d}^{k+1}$. 
If downlink $\textbf{d}^{k+1}$ is traversed by the dipath of an S-D pair $(\textbf{h}^s, \textbf{h}^d)$ in  $R^*(k+1, d)$, we have
\begin{equation}
\label{downlink_traverse_condition}
h^s_1... h^s_k h^d_{k+1}  = \textbf{d}^{k+1}\ \textup{and} \ h^d_{k+1} \neq h^s_{k+1}.
 \end{equation}
Furthermore, the number of satisified S-D pairs in $R^*(k+1, d)$ is $(d-1)d^k = d^{k+1} -d^k$ due to the following facts. (1) Downlink $\textbf{d}^{k+1}$ determines the first $k$ components of $\textbf{h}^s$, whereas the last component of $\textbf{h}^s$ has $d-1$ choices. (2) Downlink $\textbf{d}^{k+1}$ determines the last component of $\textbf{h}^d$, whereas any of the first $k$ components of $\textbf{h}^d$ has $d$ choices. 
In other words, the load of downlink $\textbf{d}^{k+1}$ is  $d^{k+1} -d^k$.
Together with (\ref{forwarding_lower}), we infer that the maximum link load in layer $k+1$ is $d^{k+1} -d^k$.

Next, we show that the maximum link load in each built-in BCube is also $d^{k+1}-d^{k} $.
First, we consider the dipath of an S-D pair $( h^s_1 ... h^s_k,   h^d_1 ... h^d_k )$ in  $R^*(k, d)$. 
The same dipath becomes a part of the dipath in $R^*(k+1, d)$ whose source and destination are given by
\begin{equation}
\small
\label{dorSet}
(h^s_1 \cdots h^s_k x, h^d_1 \cdots h^d_k y)
\end{equation}
where $x$  is the index of the  built-in BCube in $\mathcal{B}(k+1, d)$  that the source belongs to, and $y$ is the index of the built-in BCube that the destination belongs to.
The above statement applies to any dipath in $R^*(k, d)$.
Fix $y$, i.e., a built-in BCube.   
Since $x$ has a range $\{0, ..., d-1\}$, we learn that the link load in a built-in BCube is $d^{k+1} -d^k$, which is $d$ times the link load in $\mathcal{B}(k, d)$. 
Together with (\ref{forwarding_lower}), we infer that the maximum link load in a built-in BCube is $d^{k+1} -d^k$.
This completes the proof.
\end{proof}

\subsection{The Proposed RWA scheme } 
To analyze the RWA problem in BCube, we introduce a specific pattern of permutation routing called \textit{Cyclic Permutation Routing} (CPR). 
We first show that BCube provides link-disjoint dipaths for a CPR. 
We then propose an oblivious RWA scheme and derive upper bounds of the optical index. 

A permutation here is referred to as a set of S-D pairs wherein each host is a source and a destination of exactly one S-D pair. 
\begin{defn}
For  a given $\ell$-dimensional vector $p_1 ... p_\ell \in \mathbb{Z}^\ell_d$, we define a permutation, denoted by $P(p_1 ... p_{\ell})$, as follows. 
\begin{equation}
\label{permutation}
\begin{split}
P(p_1 ... p_\ell ) :=& \{(\textbf{h}^s, \textbf{h}^d): \\
& \textbf{h}^s\in \mathbb{Z}^{\ell}_d,  h^d_i  = (h^s_i +p_i)_d , i =1, 2, ..., \ell \},
\end{split}
\end{equation}
where $(x)_d  := x \mod d $, and $ h^d_i  = (h^s_i +p_i)_d $ also implies that  $p_i = (h^d_i -h^s_i)_d$ since  $h^d_i, h^s_i, \textup{and}  \ p_i \in \mathbb{Z}_d$.
\end{defn}
In particular, we call $P(0 ... 0)$ the \textit{zero} permutation where $p_i=0$ for all $i$.
\begin{lem}
\label{communication decomposition}
All host pairs in $\mathcal{B}(\ell, d)$ can be classified into $d^\ell$ CPRs with $p_1...p_\ell$ ranging over $\mathbb{Z}^\ell_d$.
\end{lem}
\begin{proof}
Given an arbitrary S-D pair $(\textbf{h}^s, \textbf{h}^d)$, it must belong to some $ P(p_1...p_\ell) $ whose $p_1, ..., p_\ell$ is equal to  
\begin{displaymath}
p_i = (h^d_i -h^s_i)_d, \ i=1, 2, ..., \ell.
\end{displaymath}
Thus, this lemma follows.
\end{proof}

\begin{defn}
For a given $P(p_1 ... p_\ell)$ in $\mathcal{B}(\ell, d)$, we define a CPR, denoted by $R(p_1 ... p_\ell )$, as follows.
\begin{displaymath}
R{(p_1 ... p_\ell )} := \{P_{\textbf{h}^s, \textbf{h}^d}: (\textbf{h}^s, \textbf{h}^d)\in P(p_1 ... p_\ell ) \},
\end{displaymath}
where $P_{\textbf{h}^s, \textbf{h}^d}$ is a descending path.
\end{defn}
In particular, dipaths in $R(p_1 ... p_{\ell-1} 0)$, where none of layer-$\ell$ links is involved, can be classified into $d$ different sets such that dipaths in each set consist of links that only belong to some built-in BCube. 
Furthermore, each of such sets is isomorphic to $R(p_1...p_{\ell-1})$, a dipath set in $\mathcal{B}(\ell-1, d)$.
Next, we show dipaths in $R{(p_1...p_\ell )}$ are link disjoint. 

\begin{lem}
\label{link-disjoint permutation routes}
Given $\mathcal{B}(\ell, d)$ and $P(p_1 ... p_\ell )$, we have dipaths in $R{(p_1 ... p_\ell )}$ are link disjoint.
\end{lem}
\begin{proof}

We proceed by induction on $\ell$. The steps are as follows.
\begin{enumerate}[(1)]
\item We first prove the result holds in $\mathcal{B}(k, d)$ when $k=1$;
\item we assume the result holds in $\mathcal{B}(k, d)$ for some $k$ in  $\{1, 2, ...., \ell-1\}$;
\item we prove the result holds in  $\mathcal{B}(k+1, d)$.
\end{enumerate}
If $\ell=1$, we have each dipath in $R(p_1)$, $p_1\neq 0$, consists of only one uplink and only one downlink; the uplink connects a source and the downlink connects a destination. 
Since a host is a source and a destination of exactly one pair in a CPR,  each directed link is traversed by only one dipath in $R(p_1)$. In other words, dipaths in $R(p_1)$ are link disjoint. 
Next, we show this result holds for $\ell = k+1$ based on the assumption for $\ell =k$. 

Consider an uplink $\textbf{u}^{k+1}$ in $\mathcal{B}(k+1, d)$.
If the dipath of an S-D pair $(\textbf{h}^s, \textbf{h}^d)$ traverses this uplink, we have
\begin{equation}
\textbf{u}^{k+1}  = \textbf{h}^s, h^s_{k+1} \neq h^d_{k+1}.
\end{equation}
Here $h^s_{k+1} \neq h^d_{k+1}$ implies that $p_{k+1} \neq 0$.
Consider a downlink $\textbf{d}^{k+1}$. 
If the dipath of an S-D pair $(\textbf{h}^s, \textbf{h}^d)$ traverses this downlink, we have
\begin{equation}
h^s_1 ... h^s_{k} h^d_{k+1}  =\textbf{d}^{k+1} \ \textup{and} \  h^s_{k+1} \neq h^d_{k+1}.
\end{equation}
According to $h^s_{k+1}  = (h^d_{k+1} -p_{k+1})_d =(d^{k+1}_{k+1} - p_{k+1})_d $ in $P(p_1 ... p_{k+1})$, we learn that  uplink $\textbf{u}^{k+1}$ (or downlink $\textbf{d}^{k+1}$) uniquely determines the source of a pair. 
We thus infer by the CPR  definition that dipaths of $R(p_1 ... p_{k+1})$ collide on neither uplinks nor downlinks in layer $k+1$. If $p_{k+1} =0$, the above statement holds naturally because dipaths in $R(p_1 ... p_{k} 0)$ do not traverse any links in layer $k+1$.
Next, we further show that dipaths in $R(p_1 ... p_{k+1})$ do not collide inside each  built-in BCube. 

Consider an S-D pair $(\textbf{h}^s, \textbf{h}^d)$ in $P(p_1 ... p_k p_{k+1})$ and a host $\textbf{h}^I$ with $\textbf{h}^I = h^s_1 ... h^s_k h^d_{k+1}$. 
If $p_{k+1} \neq 0$, the descending dipath $P_{\textbf{h}^s, \textbf{h}^d}$ arrives at host $\textbf{h}^I$ after its first hop; otherwise, $\textbf{h}^I$ is its source node. 
One can check that each $P_{\textbf{h}^s, \textbf{h}^d}$ in $R(p_1 ... p_k p_{k+1})$ has a distinct $\textbf{h}^I$, and  each descending dipath $P_{\textbf{h}^I, \textbf{h}^d}$  belongs to $R(p_1 ... p_k 0)$.
On the other hand, $R(p_1 ... p_k 0)$ can be divided into $d$ dipath sets, where each set is isomorphic to $R(p_1 ... p_k )$. Recall that we have assumed dipaths in $R(p_1 ... p_k )$ are  link disjoint. 
We can infer that dipaths in $R(p_1 ... p_k p_{k+1})$ do not collide in each built-in BCube. Thus we  finish the proof.
\end{proof}

On the basis of  Lemma \ref{link-disjoint permutation routes}, we divide RWA into two parts: path allocation and wavelength assignment. 
In the part of path allocation, we use descending dipaths only. 
In the part of  wavelength assignment, we  first indicate wavelengths by $\ell$-dimensional vectors $w_1... w_\ell$ in $\mathbb{Z}^\ell_d$, 
and then assign all dipaths in $R(p_1 ... p_\ell)$ with a single wavelength whose vector is given by
\begin{equation}
\label{RWA_1}
 w_1... w_\ell = p_1...p_\ell
 \end{equation}
Algorithm 1 illustrates more details on the proposed wavelength assignment.
It is easy to see the wavelength assignment scheme in (\ref{RWA_1}), besides its simplicity, guarantees nonblocking lightpaths for a host-to-host traffic.
\begin{algorithm}[!t]
 \caption{An oblivious RWA} 
 \KwIn{$\textbf{h}^\textup{s}, \textbf{h}^\textup{d}$} 
 \KwOut{ $w_1 .... w_\ell$}
 \For{$i=1; i\le \ell; $ ++$i$} 
 { $w_i =  (h^d_i - h^s_i)_d$\;
      }
 return $w_1 .... w_\ell$\; 
 \end{algorithm}

Since the host-to-host communication is composed of $d^\ell-1$ non-zero permutations, the scheme in (\ref{RWA_1}) uses at most $d^\ell-1$ wavelengths. 
In other words, we have $\omega(\mathcal{B}(\ell, d)) \le d^\ell-1$. 
Recall that we have $ \pi(\mathcal{B}(\ell, d)) = d^\ell -d^{\ell-1} \le \omega(\mathcal{B}(\ell, d)) $. 
Combining the two results together, we get
\begin{equation}
\label{loose_bound}
d^{\ell} - d^{\ell-1} \le \omega(\mathcal{B}(\ell, d)) \le d^\ell - 1.
\end{equation}
However, the proposed RWA  does not use wavelengths in an optimal way.
For example, dipath sets $R(1 0 0)$, $R(0 2 0)$, and $R(0 0 2)$  in  $\mathcal{B}(3, 3)$ can be assigned with a same wavelength since their dipaths use directed links of different layers. 
Towards a better understanding on the minimum usage of wavelengths, we conduct a deeper investigation on the upper bound of $\omega(\mathcal{B}(\ell, d))$ in the next section.

%


\subsection{Bounds of the Optical Index}
To derive a tigher bound of $\omega(\mathcal{B}(\ell, d))$, we transform this problem into a vertex coloring problem.
Then, we derive an upper bound of $\omega(\mathcal{B}(\ell, d))$ using existing results on the chromatic number in Graph Theory. 
To begin with, we bring the following property of CPRs, which motivates the problem transformation. 
\begin{lem}
\label{collision_analysis}
Consider $R(x_1 ... x_\ell)$ and $R(y_1 ... y_\ell)$. 
Dipaths in $R(x_1 ... x_\ell)$ and $R(y_1 ... y_\ell)$ must collide at links of layer $i$ if $x_i\neq 0$ and $y_i\neq 0$.
\end{lem}
\begin{proof}
If $x_i\neq 0$, a dipath in $R(x_1 ... x_\ell)$ must traverse an uplink and a downlink of layer $i$. Moreover, according to the link-disjoint property in Lemma \ref{link-disjoint permutation routes}, different dipaths of $R(x_1 ... x_\ell)$ use different directed links at layer $i$. 
Since the number of diapths in $R(x_1 ... x_\ell)$ is equal to that of uplinks (downlinks) at layer $i$, we infer that each uplink (downlink) of layer $i$ is traversed by exactly one dipath in $R(x_1 ... x_\ell)$. 
The above result also applies to $R(y_1 ... y_\ell)$ with $y_i\neq 0$.
The result follows. 
\end{proof}

We continue to follow the idea of assigning a single wavelength to $R(p_1 ... p_\ell)$.
To achieve nonblocking lightpaths, we assign different wavelengths to $R(x_1...x_\ell)$ and $R(y_1 ... y_\ell)$  if there exists some $i$ such that $x_i\neq 0$ and $y_i \neq 0$.
We refer to the above constraint as Wavelength Assignment Constraint (WAC), based on which, we draw a graph, denoted by $G=(R, E)$.
Each node in $R$ indicates a $R(p_1 ... p_\ell)$;  two nodes are adjacent if they satisfy WAC. 
We then use colors to represent wavelengths, and color adjacent nodes in $(R, E)$ with different colors.
Clearly, the goal of vertex coloring is to minimize the usage of colors, which is known as a vertex coloring problem.

In a vertex coloring problem, the chromatic number of a graph $\chi(G)$, is defined as  the minimum number of colors in order to color adjacent nodes with different colors.
Let $\Delta(G)$ be the maximum degree of graph $G$. 
Brooks's theorem \cite{Brooks} proved that 
\begin{equation}
\label{vertex_coloring}
\chi(G)\le \Delta(G), 
\end{equation}
where $G$ is a connected simple graph that is neither a complete graph nor an odd cycle.
Besides, it has been shown that complete graphs have
\begin{equation}
\label{vertex_coloring_2}
\chi(G)= \Delta(G)+1.
\end{equation}

Since $(R, E)$ is derived under the restriction of assigning a single wavelength to a CPR, we further have
\begin{equation}
\label{chromatic_upper_bounds}
\omega(\mathcal{B}(\ell, d)) \le \chi(G).
\end{equation}
We then prove upper bounds of $\omega(\mathcal{B}(\ell, d))$ as follows.
\begin{thm}
\label{specific case}
Let $\ell$ and $d$ be positive integers. 
Then $\omega(\mathcal{B}(1, d)) = d-1$, $\omega(\mathcal{B}(2, d)) = d^2-d$, and for $\ell \ge 3$ we have 
\begin{displaymath}
d^{\ell} - d^{\ell-1} \le \omega(\mathcal{B}(\ell, d)) \le d^\ell - d^{\lfloor \frac{\ell}{2}\rfloor}- (\lfloor \frac{\ell}{2}\rfloor-1).
\end{displaymath}
\end{thm}
\begin{proof}
We only prove the upper bound here since we have validated the lower bound in the last subsection. 
We consider respectively three cases, $\ell=1$, $\ell=2$, and $\ell >2$, and prove each case, separately. 
For the case $\ell =1$, the result holds naturally according to (\ref{loose_bound}).
For the case $\ell =2$, we assign dipaths in $R(p_1 p_2)$ with a wavelength $w_1 w_2$, whose value is given by 
\begin{equation}
\begin{cases}
w_1 = p_1, w_2 = p_2, & \textup{if}\ p_2 \neq 0,\\
w_1 = 0, w_2 = p_1, & \textup{if}\ p_2 = 0.
\end{cases}
\end{equation}
In other words, wavelength $0 x$ is used by dipaths in $R(x 0)$ and $R(0x)$ for any $x\in \{1, ..., d-1\}$.
Since dipaths in $R(x 0)$ and $R(0 x)$ use links of different layers, they can share a common wavelength.
Note that we do not include the zero permutation in host-to-host communication.
Therefore, we always have $w_2\neq 0$, which implies that the number of involved wavelengths is $d(d-1) =d^2-d$.
In other words, we have  $\omega(\mathcal{B}(2, d)) \le d^2-d$.
Together with (\ref{loose_bound}), we get $\omega(\mathcal{B}(2, d)) = d^2-d$.

The problem becomes much more complicated when  $\ell \ge 3$, where we fail to derive an exact number of the optical index.
Instead, we derive a tighter upper bound than the one in (\ref{loose_bound}). The steps are as follows.
First, we classify all CPRs into $\ell$ classes, $\mathcal{C}_1, ..., \mathcal{C}_\ell$, according to the number of non-zero components in $p_1...p_\ell$. 
That is, if $P(p_1...p_\ell)$ has $k$ nonzero components in $p_1...p_\ell$, then $P(p_1 ...p_\ell)\in \mathcal{C}_k$.
Second, we analyze the upper bound of wavelengths used by each class, separately. 
By summing up all these upper bounds together, we finally achieve an upper bound of $\omega(\mathcal{B}(\ell, d))$ for $\ell\ge 3$.

To begin with, we draw a graph, denoted by $G_k$, for each $\mathcal{C}_k$, respectively. The nodes of  $G_k$  are CPRs in $\mathcal{C}_k$, and the edges of  $G_k$ are added according to WAC.
The total number of nodes in $G_k$ is $\dbinom{\ell}{k}(d-1)^k$ since each CPR has $k$ nonzero components and  each nonzero component has $d-1$ choices.  
Let $\dbinom{\ell-k}{k} = 0 \  \textup{if} \  k>\frac{\ell}{2}$.
Consider any node in $G_k$ as a target node. 
According to WAC, we learn that there are $ \dbinom{\ell-k}{k}(d-1)^k$ nodes that are not adjacent to the target node in $G_k$.
In other words, the degree of this target node is given by
\begin{equation}
\left ( \dbinom{\ell}{k} - \dbinom{\ell-k}{k}\right)(d-1)^k -1.
\end{equation}
Besides, we can infer by symmetry that $G_k$ is a regular graph. 
Therefore, we have
\begin{equation}
\label{degree_G_k}
\Delta(G_k)  = \left( \dbinom{\ell}{k} - \dbinom{\ell-k}{k}\right)(d-1)^k-1.
\end{equation}

Note that $G_k$ with $k=1$ consists of $\ell$ independent complete graphs.
According to  (\ref{vertex_coloring_2}), the chromatic number of each independent complete graph is $d-1$. 
In particular, we have $\chi(G_1) = d-1$ since these independent complete graphs can share a common set of colors.
If  $k>\lfloor \frac{\ell}{2}\rfloor$, we also have $\chi(G_k) = \Delta(G_k) +1$ since $G_k$ is a complete graph.
For $k\in [2,\lfloor \frac{\ell}{2}\rfloor ]$, according to (\ref{vertex_coloring}), we have $\chi(G_k) \le \Delta(G_k)$.
In conclusion, we have 
\begin{equation}
\sum_{k=1}^\ell \chi(G_k) \le d-1+ \sum_{k=2}^{\lfloor \frac{\ell}{2}\rfloor}\Delta(G_k) +  \sum_{k=1+\lfloor \frac{\ell}{2}\rfloor}^{\ell} (\Delta(G_k)+1)
\end{equation}
In particular, let
\begin{equation}
\label{total number}
N_w = \sum_{k=1}^{\ell}(\Delta(G_k)+1),
\end{equation}
we have 
\begin{equation}
\sum_{k=1}^\ell \chi(G_k) \le N_w  -(\lfloor \frac{\ell}{2}\rfloor-1).
\end{equation}
By substituting  (\ref{degree_G_k}) into (\ref{total number}), we get 
\begin{equation}
\begin{split}
N_w &=  \sum_{k=1}^{\ell} \dbinom{\ell}{k}(d-1)^k  - \sum_{k=1}^{\lfloor\frac{\ell}{2}\rfloor}\dbinom{\ell-k}{k} (d-1)^k\\
& \le \sum_{k=1}^{\ell} \dbinom{\ell}{k}(d-1)^k -\sum_{k=1}^{\lfloor\frac{\ell}{2}\rfloor} \dbinom{\lfloor\frac{\ell}{2}\rfloor}{k} (d-1)^k\\
& = d^\ell -1  -(d^{\lfloor\frac{\ell}{2}\rfloor} -1) = d^\ell -d^{\lfloor\frac{\ell}{2}\rfloor}
\end{split}
\end{equation}
Therefore, we have
\begin{equation}
\sum_{k=1}^\ell \chi(G_k)  \le  d^\ell -d^{\lfloor\frac{\ell}{2}\rfloor} -(\lfloor \frac{\ell}{2}\rfloor-1).
\end{equation}
Thus, according to (\ref{chromatic_upper_bounds}), we finish the proof.
\end{proof}

\section{Conclusion}
In this paper, we study host-to-host routing in an all-optical BCube network, where we mainly focus on the forwarding and optical indices. 
Specifically, we succeed in deriving the exact values of the forwarding indices in all BCube networks and the exact values of the optical indices in BCube networks that have only one or two switch layers. 
For BCube networks with more than two switch layers, we derive tight upper bounds of the optical indices after formulating the problem as a vertex coloring problem on the basis of CPRs. 
Besides, we also propose an oblivious RWA scheme which can assign a lightpath to every S-D host pair based only on its source and destination addresses. 
Although we have shown that the proposed RWA is not optimal in wavelength usage, it has the advantage in  low implementation complexity.

\bibliographystyle{ieeetr}
\bibliography{myReference}

\end{document}